\documentclass[11pt]{article}
\usepackage{amsmath,amssymb,amsthm,setspace,graphicx}
\usepackage[text={6.5in,8.4in}, top=1.3in]{geometry}
\usepackage{algorithm}
\usepackage[noend]{algpseudocode} 

\newtheorem{theorem}{Theorem}[section]

\newtheorem{lemma}[theorem]{Lemma}

\newtheorem{definition}[theorem]{Definition}

\newcommand{\C}{{\cal C}}
\newcommand{\T}{{\cal T}}

\begin{document}

\title{A $2^{O(k)}n$ algorithm for $k$-cycle in minor-closed graph families}

\author{
Raphael Yuster
\thanks{Department of Mathematics, University of Haifa, Haifa
31905, Israel. Email: raphy@math.haifa.ac.il}
}

\date{}

\maketitle

\setcounter{page}{1}

\begin{abstract}
Let $\C$ be a proper minor-closed family of graphs.
We present a randomized algorithm that given a graph $G \in \C$ with $n$ vertices, finds a simple cycle of size $k$ in $G$ (if exists) in $2^{O(k)}n$ time.
The algorithm applies to both directed and undirected graphs. In previous linear time algorithms for this problem,
the runtime dependence on $k$ is super-exponential.
The algorithm can be derandomized yielding a $2^{O(k)}n\log n$ time algorithm.

\vspace*{3mm}

\end{abstract}

\section{Introduction}

All graphs in this paper are finite and simple. Standard graph-theoretic terminology follows \cite{bollobas-1978}.
Let $G$ be an undirected graph. A graph $H$ is a {\em minor} of $G$ if it can be obtained from $G$ by removal and contraction of edges. A family $\C$ of graphs is said to be {\em minor-closed} if a minor of a graph of the family is also a member of the family. The seminal graph minor theorem of Robertson and Seymour \cite{RS-2004} asserts that every minor-closed family of graphs can be characterized by a finite set of forbidden minors.

An undirected graph $G$ is {\em $d$-degenerate} if there is a total order $\pi$ of $V(G)$ such
that for each $v \in V(G)$, the number of neighbors of $v$ appearing in $\pi$ after $v$ is at most $d$.
The ordering $\pi$ is called a {\em $d$-degenerate ordering}.
Clearly, if $G$ is $d$-degenerate then $m\le dn$ where $n=|V(G)|$ and $m=|E(G)|$.
The smallest $d$ for which $G$ is $d$-degenerate is the {\em degeneracy} of $G$, denoted by $d(G)$.

It is well-known (and also an easy consequence of the graph minor theorem) that if $\C$ is a proper minor-closed family of graphs, i.e., a minor-closed family which is not the family of all graphs, then all graphs in $\C$ are of bounded
degeneracy. So, there exists a constant $d=d_\C$ such that every $G \in \C$ satisfies $d(G)\le d$.
In particular, all graphs in $\C$ are sparse, as they have only $O(n)$ edges.
As an example, consider the family of {\em planar graphs}. It is minor-closed and the degeneracy
of every planar graph is at most $5$ (as each planar graph has a vertex whose degree is at most $5$).

Relying on the fact that proper minor-closed graph families have bounded degeneracy, Alon et al. \cite{AYZ-1995} used the color coding method to devise a randomized linear (in $n$) algorithm  for finding simple cycles of size $k$ in directed or undirected graphs that belong\footnote{Throughout this paper, if a graph is directed, then the concepts of minor and degeneracy refer to its underlying undirected structure.} to a proper minor-closed family of graphs.
The running time of their randomized algorithm is $k^{O(k)}n$ and it can be derandomized resulting in
an $k^{O(k)}n\log n$ deterministic algorithm. An important ingredient in the proof of \cite{AYZ-1995} is to color the vertices of the graph such that a simple cycle of size $k$ will, with small probability, be colored by {\em consecutive} distinct colors. The chance of that occurring is already at most $2/k^{k-1}$, which already implies that the algorithm's dependence on $k$ is $k^{O(k)}$. To the best of our knowledge, no faster algorithm is known. Here
we present a faster linear time parameterized algorithm as the dependence on $k$ is only exponential.
We note that the exponent base is only linear in $d=d_\C$, as in \cite{AYZ-1995}.
\begin{theorem}\label{t:1}
	Let $\C$ be proper minor closed family of graphs. There is a randomized algorithm that given an $n$-vertex directed or undirected graph $G \in \C$, finds a simple directed or undirected cycle of size $k$ in $G$, if such a cycle exists, in $2^{O(k)}n$ time.
	The algorithm can be derandomized resulting in running time $2^{O(k)}n\log n$.
\end{theorem}
Following a review of related work, in Section 2 we give a high level overview of our algorithm.
Following that, in Section 3 we  set up some definitions and lemmas that are used in the randomized algorithm presented in Section 4. Derandomization is discussed in Section 4.

\subsection{Related work}
We have already mentioned the $k^{O(k)}n$ randomized algorithm of \cite{AYZ-1995} and its deterministic version (also appearing in \cite{AYZ-1995}) that incurs an additional $\log n$ factor. If linear time is not sought, then randomized color coding can find
cycles of size $k$ in any graph with $n$ vertices and $m$ edges in $2^{O(k)}mn$ time, so for proper
minor-closed graph families the runtime is $2^{O(k)}n^2$.
For the special case of planar graphs, Eppstein \cite{eppstein-1999} devised a deterministic algorithm that runs in $k^{O(k)}n$ time. In fact, Eppstein's algorithm applies to finding any pattern with $k$ vertices
(namely, subgraph isomorphism of planar graphs), not just cycles. Dorn \cite{dorn-2010} improved the dependence on $k$ for the subgraph isomorphism problem in planar graphs to $2^{O(k)}n$.
For {\em undirected} planar graphs as well as undirected apex-minor free graphs (apex graphs are graphs that can be made planar by a single vertex removal), a recent algorithm of Fomin et al. \cite{FLM+-2016} can find a cycle of size $k$ in $2^{O(\sqrt{k}\log^2 k)} \cdot n^{O(1)}$.
So here the dependence on $k$ is sub-exponential but the dependence on $n$ is not linear.
Their method relies on topological properties of planar and apex-minor free graphs (in particular, that they have locally bounded treewidth) and this property does not hold for general minor-closed families \cite{DH-2004,eppstein-2000}. Pilipczuk and Siebertz \cite{PS-2019} solve undirected subgraph isomorphism in proper minor-closed graph families in $k^{O(k)}n^{O(1)}$ using only $n^{O(1)}$ space.

The related easier problem of finding a simple $k$-path has faster algorithms.
In fact, color coding can solve the $k$-path problem in general graphs using a $2^{O(k)}m$ randomized algorithm, and improvements in the exponent, and even deterministic algorithms running in $2^{O(k)}m$ time are known.
Notable and representative results that use advanced combinatorial and algebraic techniques in
order to improve the constant in the exponent are \cite{BHKK-2017,CKL+-2009,FLPS-2016,zehavi-2015}.
 Dorn et al. \cite{DPBF-2010} solve undirected $k$-path in planar graphs in $2^{O(\sqrt{k})}n^{O(1)}$ time. Their method can also be used to answer the question ``is there a simple cycle of size {\em at least $k$}'' in the same time. Building on the theory of bidimensionality of
Demaine et al. \cite{DFHT-2005} and on the Robertson-Seymour graph minor theory, Dorn et al. \cite{DFT-2012} solve undirected $k$-path in proper minor-closed families in $2^{O(\sqrt{k})}n^{O(1)}$ time.
Although bidimensionality applies only to undirected graphs, Dorn et al. \cite{DFLRS-2013} overcome some of the obstacles encountered in the directed setting and achieve an almost sub-exponential algorithm for directed $k$-path in planar graphs and apex minor-free graphs, solving it in time $O((1+\epsilon)^k n^{f(\epsilon)}$ for any $\epsilon > 0$.

\section{Algorithm outline}

We first describe the algorithm from \cite{AYZ-1995} (hereafter ``old'' algorithm) and then outline the new algorithm and how it overcomes the obstacles encountered when trying to improve the runtime of the old algorithm.
This outline also serves as a roadmap and motivation for the definitions and steps in Section 3 and 4 that formalize
the notions discussed in the present section.

\subsection{Outline of the $k^{O(k)}n$ algorithm of \cite{AYZ-1995}}

Throughout this section we fix a proper minor closed family $\C$ and let $d=d_\C$ to be the degeneracy bound of $\C$. 
Since the input graph $G \in \C$ is $d$-degenerate, we can find in $O(dn)=O(n)$ time a {\em $d$-degenerate ordering} of its vertices (see Lemma \ref{l:construct-pi}). By this we mean a total order $\pi$ of $V(G)$ such  that for each vertex $v$, the number of neighbors of $v$ appearing after $v$ in $\pi$ is at most $d$.
Having found and fixed such an ordering $\pi$, we can label the edges of $G$ with the integers $\{1,\ldots,d\}$ such
that for each vertex $v$, the edges incident with $v$ that connect $v$ to vertices appearing after $v$ in $\pi$ have {\em distinct} labels.  We call such a labeling a {\em $d$-degenerate labeling} (see Definition \ref{d:degenerate-labeling}).

Now color the vertices of $G$ with $k$ colors, say the colors $\{0,\ldots,k-1\}$.
A (simple) $k$-cycle in $G$ is {\em well-colored} if its vertices are {\em consecutively} colored with
the colors $\{0,\ldots,k-1\}$. Observe that with probability $k^{-O(k)}$, a simple cycle $k$-cycle $C$ present in $G$ will be well-colored. Our goal is therefore to detect a well-colored cycle.
Let $G'$ be the spanning subgraph of $G$ obtained by keeping only edges
connecting consecutively colored vertices. This is done modulo $k$, so we also keep edges connecting a vertex colored $0$
and a vertex colored $k-1$. Notice that if $C$ is well-colored, then $C$ is also a cycle in $G'$.
This cleanup step of removing ``uninteresting'' edges also appears in the new algorithm (the cleanup step of
Section 4).

Suppose $C=(v_0,\ldots,v_{k-1})$ where $v_i$ has color $i$ for $i=0,\ldots,k-1$.
It could be that $v_0$ is located before $v_{k-1}$ in $\pi$, or after $v_{k-1}$ in $\pi$.
Furthermore, the edge connecting $v_{k-1}$ and $v_0$ has one of $d$ possible labels assigned by our $d$-degenerate labeling. We guess the label of this edge (there are $d$ choices)
and we guess whether $v_{k-1}$ is before $v_0$ or after $v_0$ in $\pi$ (there are two choices). We thus have a chance of $1/(2d)$ of correctly guessing the label and order. Suppose that we indeed guessed correctly and assume without loss of generality that the label is $1$ and that $v_{k-1}$ is before $v_0$.

We remove from $G'$ all edges connecting a vertex colored $k-1$ and a vertex colored $0$ {\em unless}
it is of the form $xy$ where $x$ is before $y$ in $\pi$, $x$ has color $k-1$, $y$ has color $0$ and the label of
the edges is $1$. Denoting the resulting graph by $G''$, we observe that $C$ is still a cycle in $G''$
since the cycle edge connecting $v_{k-1}$ and $v_0$ has not been removed.
An important observation is that the subgraph of $G''$ induced by the vertices colored $k-1$ and the vertices colored $0$ is a forest of rooted stars where the roots are the vertices colored $0$ and the leaves are the vertices colored $k-1$
(some of the stars may be trivial, namely isolated vertices). We next contract each such star into a single vertex and assign the color $0$ to the unified vertex of each star (so there is no longer color $k-1$ present).
The new contracted graph $G^*$ is still in $\C$, as $\C$ is minor closed. The crucial argument now is that 
there exists a well-colored cycle of size $k-1$ in $G^*$ {\em if and only if} there exists a
well-colored cycle of size $k$ in $G''$ because the only way an edge whose endpoints are colored $k-2$ and $0$
exists in $G^*$ is through a contraction of a star as above. We can now continue inductively with $G^*$
where we again compute a $d$-degnerate ordering, a $d$-degenerate labeling, and guess the order and label of the
``last'' cycle edge. The recursion bottoms when $k=3$ where we use a known $O(nd)$ algorithm to find a $C_3$ in a $d$-degenerate graph. The probability of success, namely that a $k$-cycle in $G$ survives all the way through the recursive applications is $k^{-O(k)} \cdot (1/(2d))^{k-3}=k^{-O(k)}$ hence the $k^{O(k)}$ factor in the old algorithm's runtime.

\subsection{Outline of the new $2^{O(k)}n$ algorithm}

The first obstacle when trying to improve the $k^{O(k)}n$ runtime of the old algorithm is already in the initial coloring step,
when we require the sought cycle to be well-colored, as this only occurs with probability $k^{-\Theta(k)}$.
One may try to weaken the well-colored requirement as follows. Just color the vertices with, say, the $h$ colors
$\{0,\ldots,h-1\}$ ($h$ being some small constant independent of $k$; we will see later why we must
sometimes have $h$ as large as $6$) and hope that the coloring of the $k$-cycle is {\em periodic}. That is, the colors are consecutively colored modulo $h$. Of course, for this to hold we must assume that $h$ divides $k$, so for the time being, assume for simplicity this is indeed the case, namely $k=hq$ for some integer $q$ (see Definition \ref{d:hr-coloring} of an $(h,r)$-cyclic coloring which assumes the more general case when $k=hq+r$ where
$r = k \pmod h$ is the remainder). The probability of a cycle to be periodically colored is now at least $h^{-k}=2^{-O(k)}$.

Now, instead of contracting just the ``last'' edge connecting $v_{k-1}$ and $v_0$ as we have done in the old  algorithm (recall - after guessing whether $v_{k-1}$ is before or after $v_0$ in the total order and after guessing the label of the edge connecting them) we can now contract all edges connecting vertices colored $h-1$ and vertices colored $0$, so we can contract many vertices of our periodically-colored cycle at once; not just one, but $k/h=q$.
But here we arrive at a new obstacle. There are $q$ pairs of vertices on our cycle $C=(v_0,\ldots,v_{k-1})$ having colors $h-1$ and $h$.
These are $v_{jh-1}v_{jh}$ for $j=1,\ldots,q$ (here $v_0=v_k$). But for each of these $q$ pairs, we need to guess whether $v_{jh-1}$ is before or after $v_{jh}$ in the total order, and we need to guess the label of the
edge connecting $v_{jh-1}$ and $v_{jh}$.

Let us consider label-guessing first. We would like all of the $q$ labels of the $q$ pairs to be the same (say, all of them to have the label $1$), as recall that when we do the second cleanup step (analogous to going from $G'$ to $G''$ in the old algorithm) we want to keep only edges with the same label connecting vertices colored $h-1$ and vertices colored $0$. But this is easy to achieve with $2^{-O(k)}$ probability as follows. Recall that when we construct a $d$-degenerate labeling, all we need is to assign, for each vertex $v$, distinct labels to edges connecting $v$ to (the at most $d$) vertices after it in the total order. So, instead of just assigning the distict labels arbitrarily, assign the distinct labels at random (for example, choose a random permutation of the neighbors of
$v$ appearing after $v$ in the total order and assign labels in the order dictated by the permuation). Hence, for any given edge, the probability that its label is $1$ is at least $1/d$ and the events of having label $1$ are independent for disjoint edges. So, the probability that all the $q$ edges of our $q$ pairs have the label $1$ is at least $d^{-q}=2^{-O(k)}$. This is the ``degenerate labeling step'' of Section 4.

Let us next consider order-guessing. For each vertex colored $0$ or $h-1$ we can flip a coin whether it is a ``winner'' or a ``loser''. We would hope that for each pair, $v_{jh-1}v_{jh}$, one of them is a winner, one of them is a loser, and the winner is before the loser in the total order (see Definition \ref{d:wl} of winner-loser partition). Of course, the probability that this happens is $(1/4)^q = 2^{-O(k)}$.

Once we do the cleanup (analogous to going from $G'$ to $G''$ in the old algorithm - called now the ``winner-loser cleanup step'' in Section 4) and the contraction of the vertex-disjoint stars as in the old algorithm (this is the ``contraction step'' in Section 4), we would like to claim, just as in the old algorithm, that the new graph $G^*$ has a cycle of length $k-q=(h-1)q$ if and only if the graph $G''$ before the contraction has a periodically-colored $C_k$. This certainly holds if $h \ge 4$ but this fails if $h=3$. Indeed, if $h=3$ then in $G^*$ there are cycles
with periodic coloring of $h-1=2$ colors, namely colors $\{0,1\}$, of length $2q$. But such cycles might {\em not} necessarily correspond to cycles of length $3q$ before the contraction. Indeed prior to contraction, in $G''$ vertices colored $0$ may be adjacent to vertices colored $1$ (as these are consecutive colors), hence such cycles might be {\em original} cycles of length $2q$ or just cycles where only part of the pairs are contracted, not all $q$ of them. If $h \ge 4$ this problem does not occur since in $G''$ there are no edges connecting color $0$ to color $h-2$ (as they are not consecutive colors). This means that if $h=3$ we {\em cannot} perform the contraction step.

So, one might be content with just starting with some constant $h \ge 4$, but recall that after each contraction step, we arrive at a graph where color $h$ no longer appears, and $h$-periodic cycles become $(h-1)$-periodic. So
if we start with any $h \ge 4$, already after $h-3$ rounds, we inevitably have to deal with the case $h=3$,
so now we have a new obstacle to handle as we cannot do contraction. To solve this problem we do the {\em color refinement step} of Section 4. To explain this, let $t$ be the current cycle length we are looking at (this is $t$ from Definition \ref{d:465} and Table \ref{table:1}).
In the beginning it was $k=t$, after the first contraction round it was $t=k-q=k-k/h=k(h-1)/h$, then $t=k(h-2)/h$ and so on until we arrive at the present $t$ and the case of $3$-period colorings.
So each vertex in our present graph has colors in $\{0,1,2\}$ and we are interested in detecting $3$-periodic colored cycles of size $t$, so in particular $t=0 \bmod 3$ at this point.
Assume for simplicity that $t= 0 \bmod 6$ (the case $t = 3 \pmod 6$ is similarly handled in the color refinement step of Section 4).
Each vertex of color $j \in \{0,1,2\}$ will keep its color with probability $1/2$ and change its color
to $j+3$ with probability $1/2$. So now the possible colors are $\{0,1,2,3,4,5\}$. Observe that
a $3$-periodic colored cycle of size $t$ now has a chance of at least $(1/2)^t=2^{-O(k)}$ of becoming
$6$-periodic, so after this recoloring we need to detect $6$-periodic colored cycles, namely the case $h=6$
(so at a price of probailiby $(1/2)^t$ we are again at a stage where we can do contractions).

The description above also explains why we must assume in our algorithm that $h \in \{3,4,5,6\}$
and the value of $h$ cycles through these numbers as follows: if, say, we start with $h=4$ (this is an arbitrary decision) then the next will follow from contraction leading to $h=3$, then we have to do color refienment leading to $h=6$, then contraction
leading to $h=5$, then contraction leading to $h=4$ and so on. This explains the motivation of
Definition \ref{d:465} (a 4-3-6-5 sequence).

The outline above gives a high level description of the various definitions and steps appearing in Sections 3 and 4, but assumes the idealized case where the current cycle size $t$ we are looking at divides the current $h$ we are looking at. This occurs, for example if $k=2^r$ and we start with $h=4$ as we will never have any divisibility issues to deal with.
For example, if initially $(h,k)=(4,32)$ then our following step will be contraction resulting in
$h=3$ and $t=24$, namely case $(3,24)$, then we do color refinement and go to $(6,24)$, then
contraction to $(5,20)$, then contraction to $(4,16)$, then contraction to $(3,12)$, color refinement to
$(6,12)$, then $(5,10)$, $(4,8)$, $(3,6)$, $(6,6)$, $(5,5)$, $(4,4)$ $(3,3)$ and once we arrive at this point, as in the old algorithm, the recursion bottoms and  we use a known $O(nd)$ algorithm to find a $C_3$ in a $d$-degenerate graph. If, however divisibility issues arise during this sequence, there are some technical issues to handle which motivate the more general notion of {\em type} in Definition \ref{d:465}
and Table \ref{table:1}.

Finally observe that all success probabilities that we assume throughout the algorithm are at least
$2^{-O(k)}$ and that the recursion depth is only $O(\log k)$ (this is Lemma \ref{l:prop}).
	
\section{Cyclic colorings of cycles}

Throughout the rest of this paper we fix a proper minor closed family $\C$ and fix $d=d_\C$ to be
the degeneracy bound of $\C$. Recall from the introduction that a $d$-degenerate graph has a $d$-degenerate ordering, namely a total order of its vertex set where for each vertex $v$, at most $d$ of the neighbors of $v$ appear after $v$ in the ordering.
A simple linear time algorithm that constructs a $d$-degenerate ordering $\pi$ of a $d$-degenerate graph
is well-known (see, e.g., \cite{MB-1983}). Together with the  fact that $|E(G)|\le dn$ for $G \in \C$ we obtain the following lemma.
\begin{lemma}\label{l:construct-pi}
	Let $G \in \C$ have $n$ vertices. A $d$-degenerate ordering of $G$ can be obtained in $O(dn)=O(n)$ time. \qed
\end{lemma}

It will be useful to label the edges of a graph $G \in \C$ with at most $d$ integer labels as in the following definition. 
\begin{definition}[degenerate labeling]\label{d:degenerate-labeling}
	Let $G \in \C$ and suppose that $\pi$ is a $d$-degenerate ordering of $G$.
	A labeling of $E(G)$ with $[d]=\{1,\ldots,d\}$ is a {\em $d$-degenerate labeling} if for each $v \in V(G)$, the edges
	incident with $v$ connecting it to vertices after $v$ in $\pi$ have distinct labels.
\end{definition}
Observe that given a $d$-degenerate ordering $\pi$ we can easily construct a $d$-degenerate labeling
in linear time. However, as explained in Section 2, it will be useful to assign the labels of the edges connecting $v$ to its neighbors appearing after it in $\pi$ at random (see the ``degenerate labeling step'' in Section 4).

In what follows, we always assume that our colorings are vertex colorings and the colors are taken from $Z_q=\{0,\ldots,q-1\}$ for some $q \ge 2$. For integers $a,b$ with $b \ge 2$, whenever we use
the operator $a \pmod b$, its result is the unique integer $0 \le r \le b-1$ such that $b\,|\,(a-r)$.

As explained in Section 2, given a vertex coloring, it will be important to look for simple cycles whose colorings are almost periodic, in the sense that apart from some small (or even empty) set of consecutive cycle vertices, the remaining vertices are colored periodically. We therefore require the following definition.
\begin{definition}[$(h,r)$-cyclic coloring]\label{d:hr-coloring}
	Let $h > r \ge 0$ be integers where $h \ge 3$.
	Suppose $C=(v_0,\ldots,v_{k-1})$ is a simple cycle where $k \pmod h = r$
	and that the cycle vertices are colored with $Z_{h+r}$.
	The coloring of $C$ is called {\em $(h,r)$-cyclic} if
	$c(v_i)=i \pmod h$ for~ $0 \le i < k-r$ and $c(v_{k-i})=h+r-i$ for $1 \le i \le r$.
\end{definition}
So, in an $(h,r)$-cyclic coloring of a simple cycle, the prefix of $k-r$ vertices of the cycle
is periodically colored with $0,\ldots,h-1$, and the remaining suffix of $r$ vertices is
colored with the colors $h,h+1,\ldots,h+r-1$.
For example, in a $(5,3)$-cyclic coloring of a cycle of size $13$
the sequence of colors is $0123401234567$. In an $(h,0)$-cyclic coloring (namely, if $h$ divides $k$), the entire coloring
of the cycle is periodic.

As explained in Section 2, it will be useful to guess, for a pair of consecutive colors
$j-1$ and $j$, the order in $\pi$ of a pair $u,v$ of adjacent vertices that have these colors.
Unlike \cite{AYZ-1995}, we must allow some vertices colored by $j-1$ to appear in $\pi$ before some other vertices
colored by $j$, but also allow some vertices colored by $j$ to appear in $\pi$ before some other vertices
colored by $j-1$. To facilitate this, the following definition is required.
\begin{definition}[winner-loser partition]\label{d:wl}
	Let $G$ be a graph with a coloring $c:V(G) \rightarrow Z_{q}$ and suppose that $j-1$ and $j$
	are two consecutive integers in $Z_{q}$. A partition of $c^{-1}(j-1) \cup c^{-1}(j)$ into two parts
		$(W,L)$ is called a {\em winner-loser partition}. The vertices in $W$ are {\em winners} and the vertices in $L$ are {\em losers}.
\end{definition}

A crucial definition that is used in our algorithm is a certain decreasing integer sequence.
At each iteration of the algorithm, we will search for simple cycles whose size is the current element of the sequence and which have an $(h,r)$-cyclic coloring associated with the current element in the sequence.
\begin{definition}[4-3-6-5 sequence]\label{d:465}
For every $k \ge 4$, we define a decreasing sequence starting with $k$ and ending with $3$.
Each element of the sequence is of one of $13$ {\em types} where the set of possible types is
$\{(4,3),(4,2),(4,1),(4,0),(3,1),(3,0),(6,4),(6,3),(6,2),(6,1),(6,0),(5,1),(5,0)\}$.
The first element $k$ is of type $(4,k \pmod 4)$.
Suppose the current element is $t \ge 4$. Then the next element depends on $t$ and on its type,
as defined in Table \ref{table:1}.
We call this sequence a {\em 4-3-6-5 sequence}.
The 4-3-6-5 sequence starting with $k$ is denoted by $S(k)$, it's $j$'th element is $S(k,j)$ so
$S(k,1)=k$, and its number of elements is $N(k)$ so $S(k,N(k))=3$.
\end{definition}

\begin{table}[ht]
\begin{center}
\begin{tabular}{|c|c|c|c|}
	\hline 
	current element & current type & next element  & next type \\ 
	\hline 
	start & ~ & $k$ & $(4, k \pmod 4)$ \\ 
	\hline 
	$t$ & $(4,3)$ & $t-1$ & $(4,2)$ \\ 
	\hline 
	$t$ & $(4,2)$ & $t-1$ & $(4,1)$ \\ 
	\hline 
	$t$ & $(4,1)$ & $(3t+1)/4$ & $(3,1)$ \\ 
	\hline 
	$t$ & $(4,0)$ & $3t/4$ & $(3,0)$ \\ 
	\hline 
	$t > 4$ & $(3,1)$ & $t$ & $(6,t \pmod 6)$ \\ 
	\hline 
	$t > 4$ & $(3,0)$ & $t$ & $(6, t \pmod 6)$ \\ 
	\hline 
	$t$ & $(6,4)$ & $t-1$ & $(6,3)$ \\ 
	\hline 
	$t$ & $(6,3)$ & $t-1$ & $(6,2)$ \\ 
	\hline 
	$t$ & $(6,2)$ & $t-1$ & $(6,1)$ \\ 
	\hline 
	$t$ & $(6,1)$ & $(5t+1)/6$ & $(5,1)$ \\ 
	\hline 
	$t$ & $(6,0)$ & $5t/6$ & $(5,0)$ \\ 
	\hline 
	$t$ & $(5,1)$ & $(4t+1)/5$ & $(4,1)$ \\ 
	\hline 
	$t$ & $(5,0)$ & $4t/5$ & $(4,0)$ \\ 
	\hline 
	$4$ & $(3,1)$ & $3$ & $(3,0)$ \\ 
	\hline 
\end{tabular}
\end{center}
\caption{The definition of a 4-3-6-5 sequence.}\label{table:1}
\end{table} 

\noindent
Example: We list the elements of $S(307)$, the 4-3-6-5 sequence starting at $k=307$, together with the type of each element:
$307\,(4,3)\;$
$306\,(4,2)\;$
$305\,(4,1)\;$
$229\,(3,1)\;$
$229\,(6,1)\;$
$191\,(5,1)\;$
$153\,(4,1)\;$
$115\,(3,1)\;$
$115\,(6,1)\;$
$96\,(5,1)\;$
$77\,(4,1)\;$
$58\,(3,1)\;$
$58\,(6,4)\;$
$57\,(6,3)\;$
$56\,(6,2)\;$
$55\,(6,1)\;$
$46\,(5,1)\;$
$37\,(4,1)\;$
$28\,(3,1)\;$
$28\,(6,4)\;$
$27\,(6,3)\;$
$26\,(6,2)\;$
$25\,(6,1)\;$
$21\,(5,1)\;$
$17\,(4,1)\;$
$13\,(3,1)\;$
$13\,(6,1)\;$
$11\,(5,1)\;$
$9\,(4,1)\;$
$7\,(3,1)\;$
$7\,(6,1)\;$
$6\,(5,1)\;$
$5\,(4,1)\;$
$4\,(3,1)\;$
$3\,(3,0)\;$.
Observe that in this case we have that $N(307)=35$ and, for example, $S(307,8)=115$.

Several easy observations following directly from the definition of a 4-3-6-5 sequence are that the last element (namely, $3$) is always of type $(3,0)$, every element $t$ of type $(h,r)$ satisfies $t \pmod h = r$,
and types $(4,3),(4,2)$ are only possible in the beginning of the sequence.

Notice that we can partition all but the (at most) first two elements of a 4-3-6-5 sequence into consecutive
{\em segments}. Every element of type $(4,0)$ or $(4,1)$ is the first element of a segment.
So, for example, for the sequence $S(307)$ above, the segments listed in consecutive order are
$\{305,229,229,191\}$,
$\{153,115,115,96\}$,
$\{77,58,58,57,56,55,46\}$,
$\{37,28,28,27,26,25,21\}$,
$\{17,13,13,11\}$,
$\{9,7,7,6\}$,
$\{5,4,3\}$.
\begin{lemma}\label{l:prop}
	Each segment has at most $7$ elements and all but the last segment have at least four elements.
	Furthermore, each segment has at most one element of type $(4,\cdot)$, at most one element of type $(5,\cdot)$,
	at most one element of type $(3,\cdot)$ and at most four elements of type $(6,\cdot)$.
	If $t$ is the first element of some segment, then the first element of the next segment is at most $(t+1)/2$.
	In particular, the number of segments is at most $\lfloor \log_2 k \rfloor$, the first element of the $r$'th segment is at most $\lceil k/2^{r-1} \rceil$ and $N(k) \le 7\log_2 k$.
\end{lemma}
\begin{proof}
	The claims regarding the sizes of segments and the types present in each segment follow directly from Table
	\ref{table:1}.
	Suppose now that $t$ is the first element of some segment. Then the type of $t$ is either $(4,1)$
	or $(4,0)$. Suppose first that it is of type $(4,0)$. If the size of the segment is $4$
	then the elements of the segment are precisely $\{t,3t/4,3t/4,5t/8\}$ and the next element, starting the next segment, is also of type $(4,0)$ and is $t/2$. Otherwise, the size of the segment must be $6$,
	the elements must be $\{t,3t/4,3t/4,3t/4-1,3t/4-2,(5/6)(3t/4-2)\}$, and the next element, starting the next segment, is of type $(4,1)$ and is $(4/5)(5/6)(3t/4-2) < t/2$.
	Suppose next that $t$ is of type $(4,1)$. If the size of the segment is $4$ then the elements of the segment are $\{t,(3t+1)/4,(3t+1)/4,(5t+3)/8\}$ and the next element, starting the next segment, is of type $(4,1)$ and is $(t+1)/2$. Otherwise, the size of the segment is $7$ and the elements of the segment are
	$\{t,(3t+1)/4,(3t+1)/4,(3t-3)/4,(3t-7)/4,(3t-11)/4,5t/8-51/24\}$  and the next element, starting the next segment, is of type $(4,1)$ and is $(t-3)/2 < t/2$.
	It then follows that the first element of the $r$'th segment is at least $\lceil k/2^{r-1} \rceil$, that the number of segments is at most $\lfloor \log_2 k\rfloor$ and (with room to spare) that $N(k) \le 7 \log_2 k$.
\end{proof}

Ending this section, we require a definition that associates a sequence of minors of $G \in \C$ with the elements of the 4-3-6-5 sequence $S(k)$. The goal is to facilitate the detection of simple $k$ cycles using
contractions of that cycle and $(h,r)$-cyclic colorings of the contracted cycles where $(h,r)$ corresponds to the types of the elements of $S(k)$.
\begin{definition}[$S(k)$ minor sequence]\label{d:minor-sequence}
Let $G \in \C$, let $k \ge 4$ be an integer and suppose that $N(k)=s$.
A sequence of vertex-colored graphs $G_1,\ldots,G_s$ is called an {\em $S(k)$ minor sequence of $G$}
if the following holds:
\begin{enumerate}
	\item
	$G_1=G$ and for all $1 \le i < s$, $G_{i+1}$ is a minor of $G_i$.
	\item
	For all $1 \le i \le s$, if $S(k,i)$ is of type $(h_i,r_i)$ then the vertex coloring of $G_i$ is $c_i: V(G_i) \rightarrow Z_{h_i+r_i}$.
	\item
	For all $1 \le i < s$, if $G_{i+1}$ has a simple cycle of size $S(k,i+1)$ whose coloring is
	$(h_{i+1},r_{i+1})$-cyclic, then $G_i$ has a simple cycle of size $S(k,i)$ whose coloring is
	$(h_i,r_i)$-cyclic.
\end{enumerate}
\end{definition}
Observe that for an $S(k)$ minor sequence of $G$, if it holds that $G_s$ contains a triangle
whose coloring is $(3,0)$-cyclic (namely, the colors on the triangle are $0,1,2$) then, in particular,
$G$ has a simple cycle of size $k$ (the converse, of course, does not follow from the definition). 

\section{The algorithm}

Our main result in this section is that there is a randomized algorithm such that:
\begin{enumerate}
	\item Given $G \in \C$ and integer $k \ge 4$, always constructs efficiently an $S(k)$ minor sequence of $G$.
	\item With small probability (depending only on $k$ and $d=d_\C$), if $G$ has a simple cycle of size $k$
then the last element of the sequence, $G_{N(k)}$, has a triangle that is colored $(3,0)$-cyclic.
\end{enumerate}
This, coupled with the fact that all triangles in a $d$-degenerate graph can be deterministically found
in $O(n)$ time, immediately gives the randomized algorithm claimed in Theorem \ref{t:1}.

The first lemma in this section describes the randomized algorithm satisfying item 1 above.
The lemma following that, proves the claim in item 2 above.

\begin{lemma}\label{l:main}
	Let $G \in \C$ be a graph with $n$ vertices and let $k \ge 4$ be an integer.
	There exists a randomized algorithm that constructs in $O(\log k)n$ worst case time an $S(k)$ minor sequence of $G$. This construction has the feature (proved separately in Lemma \ref{l:prob}) that if $G$ has a simple cycle of size $k$ then with probability at least $(252d)^{-k-O(\log k)}$, the last element of the sequence, $G_{N(k)}$, has a triangle that is colored $(3,0)$-cyclic.
\end{lemma}
\begin{proof}
	We construct the required $S(k)$ minor sequence sequentially, starting with $G=G_1$ and ending at $G_{N(k)}$.
	Although our construction is a randomized one, it will {\em always} be an $S(k)$ minor sequence and
	its worst-case running time is $O(\log k)n$ as it will be clear how to generate $G_{i+1}$ and its coloring
	$c_{i+1}$ from $G_i$ and its coloring $c_i$ in $O(n)$ time.
	
	{\bf Initial step}.
	As each graph $G_i$ in the sequence that we construct should, in particular, have a vertex coloring $c_i: V(G_i) \rightarrow Z_{h_i+r_i}$ where $S(k,i)$ is of type $(h_i,r_i)$, we must first define $c_1$, the vertex coloring of $G=G_1$. Let $r = k \pmod 4$ and observe that $S(k,1)=k$ is of type $(4,r)$.
	Randomly color each vertex of $G$ with a color from $Z_{4+r}$.
	
	Assume that we have already constructed $G_1,\ldots,G_i$ such that all three properties in Definition \ref{d:minor-sequence} hold. We show how to construct $G_{i+1}$.
	Note that $G_i$ is vertex colored $c_i: V(G_i) \rightarrow Z_{h_i+r_i}$ where $S(k,i)$ is of type $(h_i,r_i)$.
	For notational convenience, let $t=S(k,i)$.
	The construction of $G_{i+1}$ and $c_{i+1}$ consists of several steps, performed sequentially.
	
	{\bf Cleaning step}.
	Remove from $G_i$ all edges that cannot appear in an $(h_i,r_i)$-coloring of a simple cycle of size $t$. Suppose $v \in V(G_i)$ has color $j \in Z_{h_i+r_i}$. More formally:
	
	In the directed case we proceed as follows.
	If $1 \le j \le h_i-2$ or $h_i \le j \le h_i+r_i-2$
	then we only keep out-edges incident with $v$ of the form $(v,u)$ if $c_i(u)=j+1$
	or in-edges of the form $(u,v)$ if $c_i(u)=j-1$.
	If $j=0$ we only keep $(v,u)$ if $c_i(u)=1$ and only keep $(u,v)$ if $c_i(u) \in \{h_i-1,h_i+r_i-1\}$.
	If $j=h_i-1$ we only keep $(v,u)$ if $c_i(u) \in \{0,h_i\}$ and only keep $(u,v)$ if $c_i(u)=h_i-2$.
	If $j=h_i+r_i-1$ we only keep $(v,u)$ if $c_i(u)=0$ and only keep $(u,v)$ if $c_i(u)=h_i+r_i-2$.
	
	In the undirected case we proceed as follows.
	If $1 \le j \le h_i-2$ or $h_i \le j \le h_i+r_i-2$ then we only keep edges $vu$ where $c_i(u) \in \{j+1,j-1\}$.
	If $j=0$ we only keep $vu$ where $c_i(u) \in \{1,h_i-1,h_i+r_i-1\}$.
	If $j=h_i-1$ we only keep $vu$ if $c_i(u) \in \{0,h_i-2,h_i\}$.
	If $j=h_i+r_i-1$ we keep $vu$ if $c_i(u) \in \{0,h_i+r_i-2\}$.
	
	Notice that the above procedure is well-defined even if $r_i=0$ (there is no color $h_i$ in that case). Also, as we only remove edges, the obtained graph after cleaning, denoted by $G_i'$,
	is a minor of $G_i$. Furthermore, if $G_i$ has a simple cycle of size $t$ whose coloring is
	$(h_i,r_i)$-cyclic, then this cycle also exists in $G_i'$ as all of its edges are retained. Finally, the cleaning step can clearly be performed in $O(|E(G_i)|) \le O(n)$ time.
	
	{\bf Degenerate labeling step}. Construct a $d$-degenerate ordering $\pi$ of $G_i'$ in $O(n)$ time
	using Lemma \ref{l:construct-pi}.
	For each $v \in V$ assign to all neighbors of $v$ (in the directed case, a neighbor may an in-neighbor or an out-neighbor) that appear in $\pi$ after $v$, distinct integers in $[d]$ to obtain a $d$-degenerate labeling $\ell$ of $G_i'$. The assignment is performed at random (and the random choices made for distinct $v$ are independent). Observe that for a particular neighbor $u$ of $v$ with $\pi(v) < \pi(u)$, the probability that $\ell(v,u)=1$ is at least $1/d$, where $\ell(v,u)$ denotes the label of the edge connecting $v$ and $u$.
	
	{\bf Winner-loser step}.
	This step is done unless $t > 4$ and $(h_i,r_i) \in \{(3,0),(3,1)\}$.
	If $r_i \ge 2$ then let $j=h_i+r_i-1$ and if $r_i \le 1$ then let $j=h_i-1$.
	One exception: if $t=4$ and $(h_i,r_i)=(3,1)$ then let $j=3$.
	We call $j-1$ the {\em buffer color}.
	Consider the set $U = c_i^{-1}(j-1) \cup c_i^{-1}(j)$ of vertices of $G_i'$ whose color is either $j-1$ or $j$.
	For each $v \in U$ flip a fair coin to determine if it is a winner or a loser and obtain a winner-loser partition $(W,L)$ of $U$. This step is done in $O(n)$ time. Table \ref{table:2} designates the
	vertices on $(h,r)$-cyclic simple cycles of size $t$ that are colored with the buffer color.

\begin{table}[h!]
	\begin{center}
		\begin{tabular}{|c|c|}
			\hline 
			current type is $(h,r)$ & structure of $(h,r)$-cyclic coloring\\ 
			\hline 
			$(4,3)$ &  $0123 \cdots 0123 4\raisebox{.5pt}{\textcircled{\raisebox{-1.5pt} {5}}}6$ \\ 
			\hline 
			$(4,2)$ & $0123 \cdots 0123 \raisebox{.5pt}{\textcircled{\raisebox{-1.5pt} {4}}}5$ \\ 
			\hline 
			$(4,1)$ & $01\raisebox{.5pt}{\textcircled{\raisebox{-1.5pt} {2}}}3 \cdots 01\raisebox{.5pt}{\textcircled{\raisebox{-1.5pt} {2}}}3 4$ \\ 
			\hline 
			$(4,0)$ & $01\raisebox{.5pt}{\textcircled{\raisebox{-1.5pt} {2}}}3 \cdots 01\raisebox{.5pt}{\textcircled{\raisebox{-1.5pt} {2}}}3$ \\ 
			\hline 
			$(6,4)$ & $012345 \cdots 012345 67\raisebox{.5pt}{\textcircled{\raisebox{-1.5pt} {8}}}9$ \\ 
			\hline 
			$(6,3)$ & $012345 \cdots 012345 6\raisebox{.5pt}{\textcircled{\raisebox{-1.5pt} {7}}}8$ \\ 
			\hline 
			$(6,2)$ & $012345 \cdots 012345 \raisebox{.5pt}{\textcircled{\raisebox{-1.5pt} {6}}}7$ \\ 
			\hline 
			$(6,1)$ & $0123\raisebox{.5pt}{\textcircled{\raisebox{-1.5pt} {4}}}5 \cdots 0123\raisebox{.5pt}{\textcircled{\raisebox{-1.5pt} {4}}}5 6$ \\ 
			\hline 
			$(6,0)$ & $0123\raisebox{.5pt}{\textcircled{\raisebox{-1.5pt} {4}}}5 \cdots 0123\raisebox{.5pt}{\textcircled{\raisebox{-1.5pt} {4}}}5$ \\ 
			\hline 
			$(5,1)$ & $012\raisebox{.5pt}{\textcircled{\raisebox{-1.5pt} {3}}}4 \cdots 012\raisebox{.5pt}{\textcircled{\raisebox{-1.5pt} {3}}}4 5$ \\ 
			\hline 
			$(5,0)$ & $012\raisebox{.5pt}{\textcircled{\raisebox{-1.5pt} {3}}}4 \cdots 012\raisebox{.5pt}{\textcircled{\raisebox{-1.5pt} {3}}}4$ \\ 
			\hline 
			$(3,1)$ $t=4$ & $01\raisebox{.5pt}{\textcircled{\raisebox{-1.5pt} {2}}}3 $ \\ 
			\hline 
		\end{tabular}
	\end{center}
	\caption{An $(h,r)$-cyclic coloring of a size $t$ simple cycle where $(h,r)$ is a relevant type for the winner-loser step. Cycle vertices colored with the buffer color are circled.}\label{table:2}
\end{table} 

	{\bf Winner-loser cleanup step}.
	This step is done unless $t > 4$ and $(h_i,r_i) \in \{(3,0),(3,1)\}$.
	Remove from $G_i'$ all edges connecting two vertices is $W \cup L$ except for edges
	connecting a winner $v$ and a loser $u$, such that $\pi(v) < \pi(u)$ and $\ell(v,u)=1$.
	Denote the resulting graph by $G_i''$ and observe that $G_i''$ is a minor of $G_i$.
	This step is done in $O(n)$ time.

	{\bf Contraction step}.
	This step is done unless $t > 4$ and $(h_i,r_i) \in \{(3,0),(3,1)\}$.
	Consider the subgraph $G_i''[W \cup L]$ induced by the vertices of $W \cup L$. Then by the winner-loser cleanup
	step, this subgraph is a forest of rooted stars. Indeed, in this subgraph, each winner is incident with at most one edge (all edges of this subgraph have label $1$), all winners form an independent set and all losers form an independent set. In fact, in each such star which is not a singleton, the root is a loser and all leaves are winners.
	Singleton stars may be formed by a single isolated winner or a single isolated loser in $G_i''[W \cup L]$.
	Now, contract each star in $G_i''[W \cup L]$ to a single vertex giving the unified vertex the buffer color $j-1$
	(singleton stars that had color $j$ also receive color $j-1$). This defines the new graph $G_{i+1}$ which is a minor of $G_i$ and hence $G_{i+1} \in \C$.
	
	Notice that in $G_{i+1}$ no vertex has color $j$ anymore.
	Observe that in the case where $r_i=1$ the coloring of $G_{i+1}$ contains a gap.
	There are no vertices colored
	$h_i-1=j$ but there are still vertices colored $h_i$.
	For example, consider the case $(h_i,r_i)=(4,1)$. Then $j=3$ and the buffer color is $2$.
	Then vertices colored $2$ and $3$ spanned vertex-disjoint stars and were contracted to unified vertices having color $2$. But there are still vertices with color $4=h_i$ in $G_{i+1}$.
	To close this gap, just rename color $h_i$ to
	color $h_i-1$. Hence, the new coloring is $c_{i+1}: V(G_{i+1}) \rightarrow Z_{h_i+r_i-1}$.
	It is immediate to check Table \ref{table:1} that $h_{i+1}+r_{i+1}=h_i+r_i-1$.
	Indeed, from Table \ref{table:1} we see that the only cases where this does not hold are if $t > 4$ and
	$(h_i,r_i) \in \{(3,0),(3,1)\}$.
	
	Having defined $G_{i+1}$ and $c_{i+1}$ we have to also prove that the third condition
	of Definition \ref{d:minor-sequence} is satisfied.
	Suppose that $G_{i+1}$ has a simple cycle of size $S(k,i+1)$ whose coloring is
	$(h_{i+1},r_{i+1})$-cyclic. Let this cycle be $C=(v_0,v_1,\ldots,v_{p-1})$ where $p=S(k,i+1)$.
	We must prove that $G_i$ has a simple cycle of size $t=S(k,i)$ whose coloring is
	$(h_i,r_i)$-cyclic. In fact, we prove that already $G_i''$ has the required cycle
	and recall that $G_i''$ is a subgraph of $G_i$.
	
	We will prove this in the undirected setting (the proof for the directed setting is identical, just the notation changes from $uv$ to $(u,v)$).
	There are four cases to consider. First assume that $(h_i,r_i)$ is such that $r_i \ge 2$.
	Then we have that $(h_{i+1},r_{i+1})=(h_i,r_i-1)$, $p=t-1$, and $j-1=h_i+r_i-2$.
	The color of $v_{p-1}$ in $c_{i+1}$ is therefore $c_{i+1}(v_{p-1})=h_{i+1}+r_{i+1}-1=j-1$
	while $c_{i+1}(v_0)=0$. But observe that $v_{p-1}$ is adjacent to $v_0$ in $C$
	and that in $G_i''$ no edge colored $0$ is adjacent to a vertex colored $j-1$. So it must be that the
	star in $G_i''[W \cup L]$ that was contracted to $v_{p-1}$ contained two adjacent vertices, call them $x,y$ such that $c_i(x)=j-1$,
	$c_i(y)=j$, and $v_{p-2}x, xy, yv_0$ are all edges of $G_i''$. Notice also that $x$ and $y$ are not
	equal to any other vertex on the cycle, as the stars in $G_i''[W \cup L]$ are pairwise disjoint. 
	Hence, the cycle $C'=(v_0,v_1,\ldots,v_{p-2},x,y)$ is a simple cycle of size $p+1=t$ in $G_i''$
	and is also $(h_i,r_i)$-cyclic colored by the coloring $c_i$.
	
	Assume next that $(h_i,r_i)$ is such that $r_i=0$ (so this is possible if $h_i \in \{4,5,6\}$).
	Then we have that $(h_{i+1},r_{i+1})=(h_i-1,0)$, $p=(h_i-1)t/h_i$, and $j-1=h_i-2=h_{i+1}-1$.
	The color of all the vertices $v_w$ of $C$ where $w \pmod {h_{i+1}}=j-1$ is $c_{i+1}(v_w)=j-1$.
	Notice that there are $p/h_{i+1}$ such vertices $v_w$.
	Each such vertex $v_w$ is adjacent in $C$ to
	a vertex whose color is $0$ in the coloring $c_{i+1}$.
	But in $G_i''$ no edge colored $0$ is adjacent to a vertex colored $j-1$.
	So it must be that the star in $G_i''[W \cup L]$ that was contracted to $v_w$ contained two adjacent vertices, call them $x_w,y_w$ such that $c_i(x_w)=j-1$, $c_i(y_w)=j$, and $v_{w-1}x_w, x_wy_w, y_wv_{w+1}$
	are all edges of $G_i''$ (in the case of $w=p-1$ just define $v_{w+1}=v_0$).
	Notice also that $x_w$ and $y_w$ for any plausible $w$ are not
	equal to any other vertex on the cycle, as the stars in $G_i''[W \cup L]$ are pairwise disjoint.
	Hence, the cycle $C'=(v_0,\cdots,v_{h_i-3},x_{h_i-2},y_{h_i-2},v_{h_i-1},v_{h_i}, \cdots, v_{p-2}  
	x_{p-1},y_{p-1})$ is a simple cycle of size $p+p/h_{i+1}=t$ in $G_i''$
	and is also $(h_i,0)$-cyclic colored by the coloring $c_i$.
	
	Assume next that $(h_i,r_i)$ is such that $r_i=1$ and $h_i \in \{4,5,6\}$.
	Then we have that $(h_{i+1},r_{i+1})=(h_i-1,1)$, $p=((h_i-1)t+1)/h_i$, and $j-1=h_i-2=h_{i+1}-1$.
	The color of all the vertices $v_w$ of $C$ where $w \pmod {h_{i+1}}=j-1$ is $c_{i+1}(v_w)=j-1$.
	Notice that there are $(p-1)/h_{i+1}$ such vertices $v_w$. Each such vertex $v_w$ is either adjacent in $C$ to
	a vertex whose color is $0$ in the coloring $c_{i+1}$ or, for the next to last vertex $v_{p-2}$ (which is also
	of the form $v_w$ since $p-2 \pmod {h_{i+1}} = j-1$), it is adjacent in $C$ to $v_{p-1}$ whose color in $c_{i+1}$ is $h_{i+1}=h_i-1$, but recall that $v_{p-1}$ was just renamed to this color to close a color gap and originally $c_i(v_{p-1})=h_i$.
	But in $G_i''$ no edge colored $0$ is adjacent to a vertex colored $j-1$ and no vertex colored $h_i$ is
	adjacent to a vertex colored $j-1$.
	So it must be that the star in $G_i''[W \cup L]$ that was contracted to $v_w$ contained two adjacent vertices, call them $x_w,y_w$ such that $c_i(x_w)=j-1$, $c_i(y_w)=j$, and $v_{w-1}x_w, x_wy_w, y_wv_{w+1}$
	are all edges of $G_i''$. 
	Notice also that $x_w$ and $y_w$ for any plausible $w$ are not
	equal to any other vertex on the cycle, as the stars in $G_i''[W \cup L]$ are pairwise disjoint.
	Hence, the cycle $C'=(v_0,\cdots,v_{h_i-3},x_{h_i-2},y_{h_i-2},v_{h_i-1},v_{h_i}, \cdots, v_{p-3}  
	x_{p-2},y_{p-2},v_{p-1})$ is a simple cycle of size $p+(p-1)/h_{i+1}=t$ in $G_i''$
	and is also $(h_i,1)$-cyclic colored by the coloring $c_i$.
	For an illustrative example of this case see Figure \ref{f:contraction}.

\begin{figure}
	\includegraphics[scale=0.6,trim=46 350 134 20, clip]{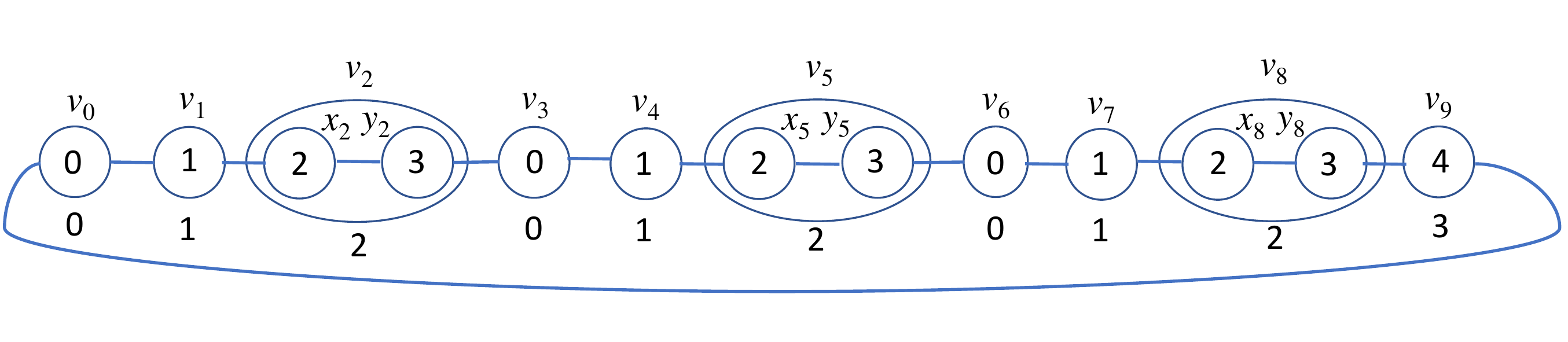}
	\caption{A simple cycle $C=(v_0,\ldots,v_9)$ of size $p=10$ in $G_{i+1}$ which is $(3,1)$-cyclic colored. The colors of $c_{i+1}$ are indicated below the vertices. It corresponds to a simple cycle of size $t=13$ in
	$G_i''$. Vertices colored $j-1=2$ in $G_{i+1}$ are results of contracted stars in $G_i''[W \cup L]$. Each such star contains two vertices $x_w,y_w$ (in this case $w \in \{2,5,8\}$) such that $c_i(x_w)=2$, $c_i(y_w)=3$
	and $v_{w-1}x_w, x_wy_w, y_wv_{w+1}$ are all edges of $G_i''$. Vertex $v_9$ has been recolored from $4$ in $c_i$ to $3$ in $c_{i+1}$ to close a color gap. The colors of $c_i$ are indicated inside the circles. The cycle in $G_i''$ is $(4,1)$-cyclic.}
	\label{f:contraction}
\end{figure}
	
	Finally consider the case $(h_i,r_i)=(3,1)$ and $t=4$. Observe that this case is identical to the case
	$(4,0)$ since for $t=4$, a $(3,1)$-cyclic coloring is identical to a $(4,0)$-cyclic coloring.
	We have already handled the case $(h_i,r_i)=(4,0)$ above.

\begin{algorithm}
	\caption{Computing an $S(k)$ minor sequence}\label{a:alg}
	\begin{algorithmic}[1]
		\Procedure{ProduceMinorSequence}{$G,k$}
		\State $G_1=G$ and $c_1: V(G_1) \rightarrow Z_{4+r}$ where $r = k \pmod 4$ constructed by the initial step.
		\For{$i=1,\ldots,N(k)-1$}
		\State $t \gets S(k,i)$
		\State Perform cleaning step on $G_i$ to obtain $G'_i$
		\State Perform degenerate labeling step on $G_i'$ to obtain a $d$-degenerate ordering and labeling
		\If{$t > 4 ~{\rm and} ~(h_i,r_i) \in \{(3,0,(3,1)\}$}
		\State $G_{i+1}=G_i$ and perform color refinement step to obtain $c_{i+1}$
		\Else
		\State Perform winner-loser step on $G'_i$ and $c_i$ to define a winner-loser partition $(W,L)$
		\State Perform winner-loser cleanup step on $G'_i$ and $W \cup L$ to obtain $G_i''$.
		\State Perform contraction step in $G_i''[W \cup L]$ to obtain $G_{i+1}$ and $c_{i+1}$.
		\EndIf
		\EndFor
		\EndProcedure
	\end{algorithmic}
\end{algorithm}

	{\bf Color refinement step}.
	This step is done only if $t > 4$ and $(h_i,r_i) \in \{(3,0),(3,1)\}$.
	In the case of type $(3,0)$ we have that $c_i : V(G_i) \rightarrow \{0,1,2\}$
	and in the case of type $(3,1)$ we have that $c_i : V(G_i) \rightarrow \{0,1,2,3\}$.
	There are four cases to consider, where in each case we do a {\em color refinement}:
	
	\noindent
	First case: $(h_i,r_i)=(3,0)$ and $t \pmod 6 = 0$.
	Each vertex of color $j \in \{0,1,2\}$ will keep its color with probability $1/2$ and change its color
	to $j+3$ with probability $1/2$.
	
	\noindent
	Second case: $(h_i,r_i)=(3,0)$ and $t \pmod 6 = 3$.
	Each vertex of color $j \in \{0,1,2\}$ will keep its color with probability $1/3$, change its color
	to $j+3$ with probability $1/3$, and change its color to $j+6$ with probability $1/3$.
	
	\noindent
	Third case: $(h_i,r_i)=(3,1)$ and $t \pmod 6 = 1$.
	Each vertex of color $j \in \{0,1,2\}$ will keep its color with probability $1/2$ and change its color
	to $j+3$ with probability $1/2$. Vertices of color $3$ will always change their color to $6$.
	
	\noindent
	Fourth case: $(h_i,r_i)=(3,1)$ and $t \pmod 6 = 4$.
	Each vertex of color $j \in \{0,1,2\}$ will keep its color with probability $1/3$, change its color
	to $j+3$ with probability $1/3$ and change its color to $j+6$ with probability $1/3$.
	Vertices of color $3$ will always change their color to $9$.
	
	\noindent
	This defines the coloring $c_{i+1}$. The graph $G_{i+1}$ will be the same as $G_i$.
	It is immediate to verify that the number of colors used in $c_{i+1}$ is $6+(t\pmod 6)$,
	so the first two conditions in Definition \ref{d:minor-sequence} hold.
	It remains to prove the third condition.
	Suppose that $G_{i+1}$ (namely, $G_i$) has a simple cycle of size $S(k,i+1)=S(k,i)=t$ whose coloring
	under $c_{i+1}$ is
	$(h_{i+1},r_{i+1})$-cyclic. Let this cycle be $C=(v_0,v_1,\ldots,v_{t-1})$.
	We must prove that the coloring of $C$ under $c_i$ is $(h_i,r_i)$-cyclic.
	Again, there are four cases to consider as in the previous paragraph.
	If $(h_i,r_i)=(3,0)$ and $t \pmod 6 = 0$ then $(h_{i+1},r_{i+1})=(6,0)$.
	So $c_{i+1}(v_w)= w \pmod 6$ for $w=0,\ldots,t-1$. But by the definition of the color refinement that
	we have done in this case, we have that $c_i(v_w)= w \pmod 3$ so $C$ under $c_i$ is $(3,0)$-cyclic.
	If $(h_i,r_i)=(3,0)$ and $t \pmod 6 = 3$ then $(h_{i+1},r_{i+1})=(6,3)$.
	So $c_{i+1}(v_w)= w \pmod 6$ for $w=0,\ldots,t-4$, $c_{i+1}(v_{t-3})=6$, $c_{i+1}(v_{t-2})=7$, 
	$c_{i+1}(v_{t-1})=8$. But by the definition of the color refinement that
	we have done in this case, we have that $c_i(v_w)= w \pmod 3$ for $w=0,\ldots,t-1$ so $C$ under $c_i$ is $(3,0)$-cyclic.
	If $(h_i,r_i)=(3,1)$ and $t \pmod 6 = 1$ then $(h_{i+1},r_{i+1})=(6,1)$.
	So $c_{i+1}(v_w)= w \pmod 6$ for $w=0,\ldots,t-2$  and $c_{i+1}(v_{t-1})=6$.
	But by the definition of the color refinement that
	we have done in this case, we have that $c_i(v_w)= w \pmod 3$ for $w=0,\ldots,t-2$ and $c_i(v_{t-1})=3$,
    so $C$ under $c_i$ is $(3,1)$-cyclic.
    Finally, if $(h_i,r_i)=(3,1)$ and $t \pmod 6 = 4$ then $(h_{i+1},r_{i+1})=(6,4)$.
    So $c_{i+1}(v_w)= w \pmod 6$ for $w=0,\ldots,t-5$, $c_{i+1}(v_{t-4})=6$, $c_{i+1}(v_{t-3})=7$, 
    $c_{i+1}(v_{t-2})=8$, $c_{i+1}(v_{t-1})=9$. But by the definition of the color refinement that
    we have done in this case, we have that $c_i(v_w)= w \pmod 3$ for $w=0,\ldots,t-2$ and $c_i(v_{t-1})=3$,
	so $C$ under $c_i$ is $(3,1)$-cyclic.
\end{proof}

The pseudocode of the algorithm defined by Lemma \ref{l:main} is given in Algorithm \ref{a:alg}.
Also observe that if a $(3,0)$-cyclic triangle exists in $G_{N(k)}$ then not only can we infer that $G$ has a simple cycle of size $k$, we can also retrace it explicitly. Indeed, the only thing needed for retracing is to mark
each edge of a contracted star in $G_i''[W \cup L]$ with the winner-loser pair that gave rise to that edge.

\begin{lemma}\label{l:prob}
	Let $G \in \C$, and suppose $G$ has a simple cycle of size $k \ge 4$.
	Then with probability at least $(252d)^{-k-O(\log k)}$ the algorithm of Lemma \ref{l:main} that constructs an $S(k)$ minor sequence of $G$ has the property that the last graph in the sequence, $G_{N(k)}$, has a triangle that is colored $(3,0)$-cyclic.
\end{lemma}
\begin{proof}
Let	$C=(v_0,\ldots,v_{k-1})$ denote the simple cycle of size $k$ assumed to exist in $G$.
Consider the coloring induced on $C$ by the coloring $c_1:V(G) \rightarrow Z_{4+r}$
constructed in the initial step of the algorithm. We lower-bound the probability that $C$ is $(4,r)$-cyclic.
By definition, this holds if $c_1(v_j)=j \pmod 4$ for~ $0 \le j < k-r$ and $c(v_{k-j})=4+r-j$ for $1 \le j \le r$
where recall that $r = k \pmod 4$ so $r \le 3$.
Hence, the probability that $C$ is $(4,r)$-cyclic under the coloring $c_1$ is at least $1/(4+r)^k \ge 7^{-k}$.
If this occurred, we say that $C$ {\em survived} in $G_1$.

Now suppose that $C$ survived in $G_i$, meaning in particular that $G_i$ has a simple cycle $C_i=(u_0,\ldots,u_{t-1})$ where $t=S(k,i)$ and whose coloring under $c_i$ is $(h_i,r_i)$-cyclic where the type
of $S(k,i)$ is $(h_i,r_i)$. Given that, we would like to lower bound the probability that $C$ survives also 
in $G_{i+1}$.

There are several cases to consider. Suppose first that $t > 4$ and $(h_i,r_i) \in \{(3,0),(3,1)\}$.
In this case, we want $C_i$ to be $(h_{i+1},r_{i+1})$-cyclic after the color refinement step that constructed
$c_{i+1}$ (recall that in this case $G_{i+1}=G_i$). Notice that in the color refinement step, vertices change their
color to another color with probability at least $1/3$. For example, in the case $(h_i,r_i)=(3,0)$ and $t \pmod 6 = 3$ vertices colored $j \in \{0,1,2\}$ change their color to one of $\{j,j+3,j+6\}$ each with probability $1/3$.
More accurately, we go over the four cases of the color refinement step.
If $(h_i,r_i)=(3,0)$ and $t \pmod 6=0$ then the probability that $C_i$ under $c_{i+1}$ is
$(h_{i+1},r_{i+1})=(6,0)$-cyclic is $(1/2)^t$.
If $(h_i,r_i)=(3,0)$ and $t \pmod 6=3$ then the probability that $C_i$ under $c_{i+1}$ is
$(h_{i+1},r_{i+1})=(6,3)$-cyclic is $(1/3)^t$.
If $(h_i,r_i)=(3,1)$ and $t \pmod 6=1$ then the probability that $C_i$ under $c_{i+1}$ is
$(h_{i+1},r_{i+1})=(6,1)$-cyclic is $(1/2)^{t-1}$.
If $(h_i,r_i)=(3,1)$ and $t \pmod 6=4$ then the probability that $C_i$ under $c_{i+1}$ is
$(h_{i+1},r_{i+1})=(6,4)$-cyclic is $(1/3)^{t-1}$.
In any case, with probability at least $(1/3)^t$, given that $C$ survived in $G_i$, $C$ also survived
in $G_{i+1}$.

Suppose next that we are not in a case ``$t > 4$ and $(h_i,r_i) \in \{(3,0),(3,1)\}$''.
What is the probability that $C$ survived after the winner-loser cleanup step?
This, in turn depends on the random choices made in the degenerate labeling step and the winner-loser step.
Again, there are several sub-cases to consider.

First assume that $(h_i,r_i)$ is such that $r_i \ge 2$. So the unique vertex on $C_i$ whose color
under $c_i$ is $j$ is vertex $u_{t-1}$ and the unique vertex on $C_i$ whose color
under $c_i$ is $j-1$ is vertex $u_{t-2}$. For $C$ to survive we must correctly guess which of
$u_{t-1}$ and $u_{t-2}$ is a winner and which is a loser, and hope that the random $d$-degenerate labeling
assigned label $1$ to the edge connecting them. Suppose $\pi(u_{t-1}) < \pi(u_{t-2})$. So for survival
we should guess that $u_{t-1}$ is a winner, $u_{t-2}$ is a loser, and $\ell(u_{t-1},u_{t-2})=1$.
This occurs with probability at least $(1/2)\cdot(1/2)\cdot(1/d)=1/(4d)$. Similarly, if $\pi(u_{t-2}) < \pi(u_{t-1})$
the probability for survival is at least $1/(4d)$. Indeed, once we have guessed correctly, the contraction
of the star in $G_i''[L \cup W]$ containing both $u_{t-2},u_{t-1}$ to a unified vertex $x$ would
create in $G_{i+1}$ a cycle $C_{i+1}=(u_0,\ldots,u_{t-3},x)$ which is colored $(h_i,r_i-1)=(h_{i+1},r_{i+1})$-cyclic under $c_{i+1}$, so $C$ survived in $G_{i+1}$.

Assume next that $(h_i,r_i)$ is such that $r_i=0$ (so this is possible if $h_i \in \{4,5,6\}$
but also the case $t=4$ and $(h_i,r_i)=(3,1)$ since this case is equivalent to $(4,0)$).
Then for every $w$ of the form $w \pmod {h_i}=h_i-1$, the vertices $u_w$ are colored with $j$ under $c_i$
and the vertices $u_{w-1}$ are colored with $j-1$ under $c_i$. In order for $C$ to survive, we would like
to correctly guess, for each such pair $u_{w-1},u_w$ the winner, the loser, and that the random label between them
is $1$. As in the previous paragraph, this occurs with probability at least $1/(4d)$ for each such pair, so the
probability that $C$ survived in $G_{i+1}$ is at least $(1/4d)^{t/h_i}$.
Indeed, once we have guessed correctly, the contraction
of the stars in $G_i''[L \cup W]$ containing both $u_{w-1},u_{w}$ for each of the $t/h_i$ plausible $w$
creates a unified vertex $x_w$ for each of them and hence there is a cycle
$C_{i+1}=(u_0,\cdots,u_{h_i-3},x_{h_i-1},u_{h_i},\cdots,u_{t-h_i},\cdots,u_{t-3},x_{t-1})$ 
which is colored $(h_i-1,0)=(h_{i+1},r_{i+1})$-cyclic under $c_{i+1}$.

Finally, the case that $(h_i,r_i)$ is such that $r_i=1$ and $h_i \in \{4,5,6\}$ is proved in exactly the same way as the previous one. Just observe that the number of plausible $w$ in this case is $(t-1)/h_i$
so the probability that $C$ survived in $G_{i+1}$ is at least $(1/4d)^{(t-1)/h_i}$.

Let us now multiply all of the lower bounds of the probabilities of survival in each iteration, to obtain
a lower bound for the survival probability of $C$ in the final graph $G_{N(k)}$, meaning that 
$G_{N(k)}$ has a triangle that is colored $(3,0)$-cyclic.
Let $p_i$ denote the probability of survival in $G_i$.
Then, summarizing what we have just proved:
\begin{enumerate}
\item[(i)] $p_1 \ge 7^{-k}$.
\item[(ii)] If $t=S(k,i) > 4$ is of type $(3,0)$ or $(3,1)$ then $p_i \ge (1/3)^t$.
\item[(iii)] Otherwise, if $t=S(k,i)$ is of type $(h_i,r_i)$ with $r_i \ge 2$ then $p_i \ge 1/(4d)$.
\item[(iv)] Otherwise, $p_i \ge (1/4d)^{t/h_i}$.
\end{enumerate}

The cases of types $(4,3)$ and $(4,2)$ are only possible at the beginning. Namely, if $k  \pmod 4 = 3$
then $S(k,1)$ is of type $(4,3)$ and $S(k,2)$ is of type $(4,2)$. If $k  \pmod 4 = 2$
then $S(k,1)$ is of type $(4,2)$. We never return to these types anymore.
So, the product of the $p_i$'s until the first time we reach the head of a segment (recall Lemma \ref{l:prop})
is either $p_1$ in the case where $k \pmod 4 \in \{0,1\}$, or $p_1/(4d)$ in the case $k \pmod 4 = 2$
or $p_1/(16d^2)$ in the case $k \pmod 4 = 3$.  We next compute the product of the $p_i$'s corresponding to the elements of some segment whose first element is $t$.
In every segment there is at most one element of type $(3,0)$ or $(3,1)$.
Hence the contribution of this element to the product of the $p_i$'s of the segment is
at least $(1/3)^t$. There is at most one element whose type is in $\{(6,1),(6,0)\}$,
at most one element whose type is in $\{(5,1),(5,0)\}$, and at most one element whose type is in $\{(4,1),(4,0)\}$.
Hence, their contribution to the product is at least $(1/4d)^{t/4}(1/4d)^{t/5}(1/4d)^{t/6}$
(we could have further optimized the exponent since, e.g., the element of type $(6,1)$ or $(6,0)$ is already at most $\lceil 3t/4 \rceil < t$ but we do not worry about optimizing the base of the exponent here).
Finally, there are at most three elements of type in $\{(6,4),(6,3)(6,2)\}$ so their contribution to the product is at least $(1/4d)^{3}$.
Overall, the product of the $p_i$'s of a segment whose first element is $t$ is at least
$$
\left(\frac{1}{3}\right)^t\left(\frac{1}{4d}\right)^{3+t/4+t/5+t/6}=
\frac{1}{64d^3}\left[\frac{1}{3}\left(\frac{1}{4d}\right)^{37/60}\right]^t
$$
But recall from Lemma \ref{l:prop} that the first element of the $r$'th segment is at most $\lceil k/2^{r-1} \rceil$
and that there are at most $\lfloor \log_2 k \rfloor$ segments. 
It follows that the probability that $C$ survived until the last element $G_{N(k)}$ is at least
\begin{align*}
& \frac{p_1}{16d^2} \prod_{r=1}^{\lfloor \log_2 k \rfloor}  \frac{1}{64d^3}\left[\frac{1}{3}\left(\frac{1}{4d}\right)^{37/60}\right]^{\lceil k/2^{r-1} \rceil}\\
&
\ge \frac{7^{-k}}{16d^2}\left(\frac{1}{64d^3}\right)^{\log_2 k}\left[\frac{1}{3}\left(\frac{1}{4d}\right)^{37/60}\right]^{\log_2 k}\left[\frac{1}{3}\left(\frac{1}{4d}\right)^{37/60}\right]^{2k} \\
& = \frac{1}{16d^2}\left(\frac{1}{64d^3}\right)^{\log_2 k}\left[\frac{1}{3}\left(\frac{1}{4d}\right)^{37/60}\right]^{\log_2 k}\left[\frac{1}{63}\left(\frac{1}{4d}\right)^{37/30}\right]^{k} \\
& = \left[\frac{1}{63}\left(\frac{1}{4d}\right)^{37/30}\right]^{k+O(\log k)}
\end{align*}
In fact, we can do a bit better since it is immediate from the proof of the lemma that each introduction of a term $1/(4d)$ in the probability expression corresponds to an edge contraction of the cycle $C$. As the overall number of contractions from a cycle of size $k$ until a cycle of size $3$ is reached at the final iteration
is less than $k$, the probability above can be  improved to at most
$$
\left[\frac{1}{63}\left(\frac{1}{4d}\right)\right]^{k+O(\log k)} = \left(\frac{1}{252d}\right)^{k+O(\log k)}\;.
$$
\end{proof}

\section{Derandomization}

It is not difficult to see that the number of ``random bits'' that we use throughout the algorithm is $O(k)$
and that the whole algorithm is encoded with a binary string of length $n$. Hence it is fairly standard to use
the derandomization method of ``almost $k$-wise independent random variables'' for our purposes.
In what follows we make this argument precise.

Consider the sequence $S(k)$ and recall that $S(k,i)$ is the $i$'th element in the sequence and
that its type is $(h_i,r_i)$.
In order to derandomize our algorithm, we first need to define certain vector-valued random variables
$M_0,\ldots,M_{N(k)-1}$, one for each but the last element of the sequence $S(k)$, and $M_0$ corresponding to the
initial step.
The length of each vector $M_i$ is $n$ (note: we do not know a priori how many vertices would be in each $G_i$ but we do know that there are never more than $n$).

We now define our sample space, namely the possible entries of each coordinate of each vector.
For $M_0$, each coordinate can be an element of $Z_{4+r}$ where $r = k \pmod 4$.
For $M_i$, consider the type $(h_i,r_i)$ of $t=S(k,i)$.
If $(h_i,r_i) \in \{(3,0),(3,1)\}$ and $t > 4$, then each coordinate can be one of $\{0,1,2\}$.
Otherwise, each coordinate is an element of $\{0,1\} \times [d]$.

An instantiation of the random variables $M_0,\ldots,M_{N(k)-1}$ exactly defines the behavior of our randomized algorithm, as follows. First, let us fix a labeling of the $n$ vertices with distinct integers from $[n]$.
We will use this labeling throughout in all graphs $G_i$, since if a star is contracted at some point
then the unified vertex can be labeled, say, by the smallest label of a vertex in the star.
Hence for every vertex $v$, and for every graph $G_i$ throughout the algorithm, the entry $M_i[v]$ is well defined
(it is just the coordinate of $M_i$ which equals the label of $v$ in $G_i$).

For the initial step, recall that we randomly color the vertices of $G$ with a coloring
$c_1: V(G) \rightarrow Z_{4+r}$. So, each vertex $v$ is colored by the color $M_0[v]$.

Now suppose we are at iteration $i$ where we have the graph $G_i$ and its coloring $c_i$.
We proceed as in Lemma \ref{l:main}. If $t=S(k,i) > 4$ and $(h_i,r_i) \in \{(3,0),(3,1)\}$
then we have to perform the color refinement step. Recall that in this step, every vertex changes its color to one of two or three possible other
colors. For example, in the case of type $(3,0)$ and $t \pmod 6 = 0$ a vertex colored $j \in \{0,1,2\}$ either keeps its color or changes its color to $j+3$. So, looking at $M_i[v]$, if $M_i[v]=0$ we do not change the color,
if $M_i[v] =1$ we change the color to $j+3$ and if $M_i[v]=2$ we can decide either way.
Otherwise, recall from Lemma \ref{l:main} that we choose for certain vertices (those colored with the
buffer color or those color with the next color after the buffer color) whether it is a winner or a loser.
Also, for every winner, we choose one of its at most $d$ incident edges connecting it to vertices appearing
after it in $\pi$ the label $1$. So, for each such vertex $v$ for which we need to decide winner/loser,
we examine $M_i[v]=(x,y) \in \{0,1\} \times [d]$. If $x=0$ it is a winner, if $x=1$ it is a loser.
If it is a winner we label the $y$'th edge connecting it to a vertex appearing after $v$ in $\pi$ with the label $1$ (the ordering of the neighbors appearing after $v$ in $\pi$ is set to be the order of the labels of these vertices). Notice that it can be that $y$ is larger than the number of vertices appearing after $v$ in $\pi$, in which case we don't label any edge incident with $v$ with the label $1$.
We have completely defined the execution path of the algorithm as a result of the values of
$M_0,\ldots,M_{N(k)-1}$. Stated otherwise, given $M_0,\ldots,M_{N(k)-1}$, the algorithm of Lemma \ref{l:main} is completely deterministic.

We would like to explicitly find a small set $\T$ of instantiations of the $M_0,\ldots,M_{N(k)-1}$ 
such that we are {\em guaranteed} that a simple cycle $C$ of size $k$ survives throughout all the iterations,
as in the proof of Lemma \ref{l:prob}. What do we then require from our set $\T$?
As for $M_0$, we require that every set of $k$ vertices (i.e. coordinates of $M_0$) will receive
any possible coloring in $Z_{4+r}$.
As for $M_i$ when $(h_i,r_i) \in \{(3,0),(3,1)\}$ and $t=S(k,i) > 4$, we would like every set of $t$
coordinates of $M_i$ to obtain all $3^t$ possible values of $\{0,1,2\}$.
As for the remaining $M_i$ we would like every set of $t$ coordinates of $M_i$ (note: this is more than needed,
if $r_i \ge 2$ then there are just two vertices on the surviving cycle that should be declared winners or losers
and if $r_i \in \{0,1\}$ there are at most $2t/h_i$ vertices on the surviving cycle that should be declared winners or losers) to obtain every possible value of  $\{0,1\} \times [d]$ (there are $(2d)^t$ such options).

It would be more convenient to view the $M_i$ as binary vectors. So, in $M_0$, only three bits are
enough to describe the entry $M_0[v] \in Z_{4+r}$ since $4+r \le 7$. So the length of $M_0$ is $3n$ bits.
For $M_i$ corresponding to types $(h_i,r_i) \in \{(3,0),(3,1)\}$ with $t \ge 4$, it suffices to use two bits
for each entry as the entries are in $\{0,1,2\}$. For the remaining $M_i$, they contain entries from
$\{0,1\} \times [d]$ so $1+ \lceil \log_2 d \rceil$ bits suffice for each coordinate.

So we would like our set $\T$ to have instantiations such that for every $3k$ bits from $M_0$, every $2t$ bits from $M_i$ corresponding to types
$\{(3,0),(3,1)\}$ and $t > 4$, and every $t(1+ \lceil \log_2 d \rceil)$ bits from the remaining $M_i$,
all possible choices are present. What is the total sum of the number of bits that we are considering?
By Lemma \ref{l:prop}, $t$ decreases by a half after each segment, so overall we are examining at most
$O(k\log d)=O(k)$ bit locations.
Viewing the $M_0,\ldots,M_{N(k)-1}$ as a consecutive sequence of binary vectors, its length is
$O(\log k n)$, so what we are looking for in $\T$ is a set of binary vectors of the same length
$N=O(\log k n)$ each,
such that for every choice of $\ell=O(k)$ bit locations, and for any choice of the $2^\ell$ values in these locations, there will be a vector in $\T$ which, when projected to these locations, yields these values.
In other words, we need a sequence $X_1,\ldots,X_N$  of random Boolean variables that are $(2^{-\ell},\ell)$-independent.
For this purpose, it suffices to use the well-known construction of Alon et. al. \cite{AGHP-1992}
(see also Naor and Naor \cite{NN-1993}).
In this construction the size of $\T$ is only $2^{O(\ell)}\log N$ and the time it takes to construct them is
only  $2^{O(\ell)}N\log N$. In our case, the size of $\T$ is therefore $2^{O(k)}\log n$ and the time to construct it is
$2^{O(k)}n\log n$. We have therefore shown how to derandomize our algorithm and obtain a worst-case running time of
	$2^{O(k)}n\log n$, as required. \qed

\section*{Acknowledgmet}

I thank the reviewers for their comments leading to an improved exposition of the paper.

\bibliographystyle{plain}

\bibliography{references}

\end{document}